\documentclass[conference]{IEEEtran}

\usepackage[bookmarks, colorlinks=true, allcolors=blue]{hyperref}
\usepackage{amssymb}
\usepackage{amsmath}
\usepackage{amsthm}
\usepackage{bbm}
\usepackage[utf8]{inputenc}
\usepackage{dsfont}
\usepackage{mathtools}

\usepackage{centernot}

\def\openone{\leavevmode\hbox{\small1\kern-3.8pt\normalsize1}}

\def\11{\mathbb{I}}

\def\sln{\succeq_{\operatorname{l.n.}}}
\def\mc{\succeq_{\operatorname{m.c.}}}
\def\deg{\succeq_{\operatorname{deg}}}

\newtheorem{definition}{Definition}[section]
\newtheorem{proposition}[definition]{Proposition}
\newtheorem{lemma}[definition]{Lemma}

\newtheorem{theorem}[definition]{Theorem}

\newtheorem{observation}[definition]{Observation}

\newtheorem{example}[definition]{Example}
\newtheorem{remark}[definition]{Remark}

\newcommand{\sign}[1]{{\rm sign}(#1)}

\usepackage{pst-node}
\usepackage{tikz-cd}

\newcommand{\cL}{{\cal L}}

\newcommand{\cY}{{\cal Y}}

\def\ln{\succeq_{\operatorname{l.n.}}}

\usepackage{graphicx}
\usepackage{setspace}
\usepackage{verbatim}
\usepackage{subfig}

\numberwithin{equation}{section}
\newcommand{\abs}[1]{\lvert#1\rvert}
\DeclareRobustCommand\openone{\leavevmode\hbox{\small1\normalsize\kern-.33em1}}

\newcommand{\be}{\begin{equation}}
	\newcommand{\ee}{\end{equation}}
\newcommand{\bea}{\begin{eqnarray}}
	\newcommand{\eea}{\end{eqnarray}}
\newcommand{\beas}{\begin{eqnarray*}}
	\newcommand{\eeas}{\end{eqnarray*}}

\DeclareFontFamily{U}{mathx}{\hyphenchar\font45}
\DeclareFontShape{U}{mathx}{m}{n}{<-> mathx10}{}
\DeclareSymbolFont{mathx}{U}{mathx}{m}{n}
\DeclareMathAccent{\widebar}{0}{mathx}{"73}

\newcommand{\cX}{{\cal X}}

\DeclareMathAccent{\widehat}{0}{mathx}{"70}
\DeclareMathAccent{\widecheck}{0}{mathx}{"71}

\begin{document}

\title{Partial orders and contraction for BISO channels}
\author{
\IEEEauthorblockN{Christoph Hirche\textsuperscript{\textsection}}
\IEEEauthorblockA{Institute for Information Processing (tnt/L3S),\\ Leibniz Universit\"at Hannover, Germany}
\and
\IEEEauthorblockN{Oxana Shaya\textsuperscript{\textsection}} 
\IEEEauthorblockA{Institute for Information Processing (tnt/L3S),\\ Leibniz Universit\"at Hannover, Germany}

}

\maketitle
\begingroup\renewcommand\thefootnote{\textsection}
\footnotetext{The order of authors is chosen alphabetically.}
\endgroup

\begin{abstract}
A fundamental question in information theory is to quantify the loss of information under a noisy channel. Partial orders and contraction coefficients are typical tools to that end, however, they are often also challenging to evaluate. For the special class of binary input symmetric output (BISO) channels, Geng et al. showed that among channels with the same capacity, the binary symmetric channel (BSC) and binary erasure channel (BEC) are extremal with respect to the more capable order. Here, we show two main results. First, for channels with the same KL contraction coefficient, the same holds with respect to the less noisy order. Second, for channels with the same Dobrushin coefficient, or equiv. maximum leakage or Doeblin coefficient, the same holds with respect to the degradability order. In the process, we provide a closed-form expression for the contraction coefficients of BISO channels. We also discuss the comparability of BISO channels and extensions to binary channels in general.
\end{abstract}

\section{Introduction}\label{sec:intro}
Data processing inequalities state that information measures are monotone under the application of a noisy channel, i.e. information can only decrease along the transmission, qualifying them as distinguishability measures~\cite{polyanskiy2016}. Contraction coefficients quantify by how much the measure decreases after the application of the channel and hence provide a strengthening of the data processing inequality. A closely related concept are channel partial orders, formalizing the idea that one channel can be generally more suitable to transmit information by some measure. An example of this connection was given by Polyanskiy et al.~\cite[Proposition 15]{polyanskiy2016} who characterized the Kullback–Leibler (KL) contraction coefficient via the maximal erasure probability $\beta$ such that the erasure channel dominates the given channel with respect to the less noisy partial order~\cite{korner1975comparison},
\begin{align}\label{Eq:KL-LN-con}
    \beta(P) &= \sup\{\varepsilon : BEC(\varepsilon) \ln P_{Y|X}\} \\
      \eta_{KL}(P) &= 1- \beta(P)
\end{align}
For technical definitions, we refer to the next section. 
Using instead the degradation partial order~\cite{bergmans1973random} in the right-hand side of Equation~\eqref{Eq:KL-LN-con}, leads to the definition of Doeblin coefficents~\cite{doeblin1937proprietes}. 
Yet another partial order is the more capable order~\cite{korner1975comparison}. 
Geng et al.~\cite{geng2013broadcast} established that among all binary input symmetric output (BISO) channels with the same capacity the binary erasure channel (BEC) and the binary symmetric channel (BSC) form the two extremes with respect to the more capable partial order. This leads to their main result that the inner and outer bounds for the corresponding BISO broadcast channel differ if and only if the two BISO channels are more capable comparable. 
In a similar line of work, \cite{Makur_2018} derived a criterion for the less noisy domination by a symmetric channel which in turn implies a logarithmic Sobolev inequality for the original channel. Beyond those results, partial orders and contraction coefficients have found numerous further applications~\cite{asoodeh2020privacy,hirche2022bounding,xu2016information}.  
This work identifies the BSC and BEC as extremal for BISO channels of the same contraction coefficient, resp. Doeblin coefficient w.r.t. the less noisy resp. more capable partial order and also comments on general binary channels.

\section{Preliminaries}\label{sec:pre}

In this work, we mainly consider binary input symmetric output (BISO) channels with input alphabet $\cX=\{0,1\}$ and output alphabet $\cY=\{0,\pm 1, \pm 2,\dots,\pm l\}$ for some integer $l\geq1$, those are the channels for which $P_{Y|X}(y|0)=P_{Y|X}(-y|1)\coloneqq p_y$. We can always assume that the output alphabet $\cY$ has even number of elements because we can split $Y=0$ into two outputs, $Y = 0_+$  and $Y= 0_-$, with $P_{Y|X}(0_-|0) = P_{Y|X}(0_+|0) = \frac{p_0}{2}$. 
The binary convolution is denoted by $a\ast b=a(1-b)+(1-a)b$ and the binary entropy function by $h_2(p)= - p\log_2(p)-(1-p) \log_2(1-p)$. We denote the  Bernoulli distribution with probability $P(X=0) = p $ by $Ber(p)$.
We are interested in ordering channels for which we employ the following partial orders. We say $P_{Y|X}$ is \textit{more capable} than $Q_{Y'|X}$, denoted $P \mc Q$, if for all $P_X$, 
\begin{align}
    I(X:Y) \geq I(X:Y'). 
\end{align}
Here, $I(X:Y)$ denotes the mutual information. 
$P_{Y|X}$ is \textit{less noisy} than $Q_{Y'|X}$, denoted $P \ln Q$, if for all $P_{UX}$, s.t. $U-X-Y$ (resp. $U-X-Y'$), 
\begin{align}
    I(U:Y) \geq I(U:Y'). 
\end{align}
Finally, $P_{Y|X}$ is \textit{degradable} into $Q_{Y'|X}$, denoted $P \deg Q$, if there exists a channel $D_{Y'|Y}$ such that,
\begin{align}
    Q_{Y'|X} = D_{Y'|Y} \circ P_{Y|X}. 
\end{align}
Equivalently, the degradable partial order is characterised via the conditional min-entropy \cite{buscemi2017},
\begin{align}\label{eq:deg_min_entropy}
   \forall P_{UX},\; H_{\min}(U|Y) \leq H_{\min}(U|Y'). 
\end{align}
We will further group the BISO channels into classes for which we need some information quantifiers. The capacity of a channel $P_{Y|X}$ is given by 
\begin{align}
    C(P_{Y|X}) = \sup_{P_X} I(X:Y). 
\end{align}
The KL contraction coefficient and the Dobrushin coefficient are respectively defined by, 
\begin{align}
    \eta_{KL}(P) &= \sup_{P_X,Q_X} \frac{D(P\circ P_X\|P\circ Q_X)}{D(P_X\|Q_X)}, \\
    \eta_{TV}(P) &= \sup_{P_X,Q_X} \frac{TV(P\circ P_X,P\circ Q_X)}{TV(P_X,Q_X)}, 
\end{align}
where we denote the distribution on $Y$ from applying $P$ to $P_X$ resp. $Q_X$ by  $P\circ P_X$ resp. $P \circ Q_X$. We will use frequently that $\eta_{KL}(P)=\eta_{\chi^2}(P)$, where the latter is the contraction coefficient for the $\chi^2$-divergence, see~\cite{choi1994equivalence}. 

If we replace the supremum with an infimum we call them expansion coefficients and denote them with $\widecheck\eta$. They quantify the preservation of information after the application of a noisy channel. 
The Doeblin coefficient 
\begin{align}
     \alpha( P_{Y|X}) =\sum_{y \in \mathcal{Y}} \min_{x} P_{Y|X}(y|x)
\end{align}
is also characterised by the extremal erasure probability such that the given channel is a degraded version of an erasure channel,
\begin{align}\label{eq:Doeblin_E}
      \alpha( P_{Y|X}) = \sup\{\varepsilon : BEC(\varepsilon) \deg P_{Y|X}\}.
\end{align}
Maximal leakage~\cite{issa2019operational} is the logarithm of the max-Doeblin coefficient introduced in \cite{Makur_2024},
\begin{align}
    \cL(X\rightarrow Y) &= \log \alpha_{\max}(P), \\
    \alpha_{\max}(P) &= \sum_{y \in \mathcal{Y}} \max_{x} P_{Y|X}(y|x).
\end{align}
Here, we slightly deviate from the usual definition of maximal leakage by making it a quantity of the channel alone assuming that the input has full support. 
In this work, we are mostly focused on binary input channels. For those, many of the above quantities are directly related to each other. 
\begin{observation}\label{ob:bin}
   For any binary channel $P_{Y|X}$, we have,
   \begin{align}
       \eta_{TV}(P) &= TV(P_{Y|X=0}, P_{Y|X=1})  = 1-\alpha(P)\\
       &= \alpha_{\max}(P)-1 =  e^{\cL(X\rightarrow Y)}-1 =\widecheck\eta_{TV}(P) .        
   \end{align}
\end{observation}
\begin{proof}
The first equality follows from the equivalent representation of the Dobrushin coefficient that is $\eta_{TV}(P) = \sup_{x,x'} TV(P_{Y|X=x},P_{Y|X=x'})$.
Most of the remaining equalities follow either, as noted in \cite{Makur_2024}, from the affinity characterisation of the TV distance $TV(P,Q) =1 - \sum_y \min\{P(y), Q(y)\}= \sum_y \max\{P(y), Q(y)\} - 1 $, or by definition. The final equality is a direct calculation. 
\end{proof}
Specializing further on the aforementioned BISO channels, also the KL contraction coefficient simplifies and we can give a closed-form expression. 
\begin{lemma}\label{lem:etaKL}
   For any BISO channel $P_{Y|X}$  we have
   \begin{align}
       \eta_{KL}(P) &= \sum_{y>0} \frac{(p_y-p_{-y})^2}{p_y+p_{-y}}
   \end{align}
\end{lemma}
\begin{proof}
We consider the binary input distributions $P_X = Ber(p) $ and $Q_X = Ber(q)$.
Note that for any binary channel, by direct calculation, the contraction coefficient simplifies to 
\begin{align} \label{eq:eta_bin}
    \eta_{KL}(P) = \eta_{\chi^2}(P) 
    =
    \sup_{q} \sum_{y} \frac{(p_{y|x=0}-p_{y|x=1})^2(1-q)q}{qp_{y|x=0} +(1-q)p_{y|x=1}}.
\end{align}
For a BISO channel, this evaluates to 
        \begin{align}
      &\sup_{q}\sum_y (p_y-p_{-y})^2 \left( 
    \frac{q(1-q)}{qp_y+(1-q)p_{-y}}
    \right) \label{eq:unoptimised_contraction}
    \\&=
    \sup_{q}\sum_y \frac{(p_y-p_{-y})^2}{p_y+p_{-y}} \frac{q(1-q)}{q \ast \frac{p_{-y}}{p_y+p_{-y}}} \\&=
   \sup_{q} \sum_{y>0} \frac{(p_y-p_{-y})^2}{p_y+p_{-y}} 
   \frac{q(1-q) }{q \ast \frac{p_y}{p_y+p_{-y}} (1- q \ast \frac{p_y}{p_y+p_{-y}}) }. 
    \end{align}
    We can exchange the supremum and the sum if every summand is optimized by the same $q$. 
    Now we show that $f(q) =  \frac{q(1-q) }{q \ast \frac{p_y}{p_y+p_{-y}} (1- q \ast \frac{p_y}{p_y+p_{-y}}) }$ has its maximum at $q= \frac{1}{2}$ for any $p_y$. Let us denote $\delta_y  = \frac{p_y}{p_y+p_{-y}}$.
\begin{align}
    &\frac{df(q)}{dq}  = \frac{ - \delta_{-y} \delta_y (2 q-1)}{(q \ast \delta_y )^2 (q\ast \delta_{-y})^2}  \overset{!}{=} 0 \\
    \Leftrightarrow\, &(\delta_y \in \{0,1\} \land q \not\in \{0,1\}) \lor q = \frac{1}{2}.\\
     \frac{d^2f(q)}{dq^2}  &= -2\delta_{-y}\delta_y\left( \frac{1}{(q \ast \delta_y)^3} + \frac{1}{(q \ast \delta_{-y})^3} \right), \\
     \frac{d^2f(\tfrac12)}{dq^2}  &= - 32 \delta_{-y}  \delta_y \leq 0.
\end{align}
Hence $q= \frac{1}{2}$ is a maximum and the result follows.
\end{proof}

As we are interested in extrema among channels with common properties, we define the following classes: 
\begin{align}
    C_C &= \{ P_{Y|X}  \mid C(P) = c \}, \\
    C_\eta &= \{ P_{Y|X}  \mid \eta_{KL}(P) = c \}, \\
    C_\alpha &= \{ P_{Y|X}  \mid \eta_{TV}(P) = c \}.
\end{align}
We are now ready to discuss our main results in the following section. 
\section{Main Results}

Let $F(C)$ denote an arbitrary BISO channel belonging to the class $C$. Abusing notation, we denote by $BSC(C)$ and $BEC(C)$ the binary symmetric channel
and the binary erasure channel belonging to the respective class.

\begin{theorem}\label{th:eta}
    For any BISO channel $F(C_\eta)$, we have
    \begin{align}
        BEC(C_\eta) \ln F(C_\eta) \ln BSC(C_\eta). 
    \end{align}
\end{theorem}
\begin{proof}
By Equation~\eqref{Eq:KL-LN-con} we directly have $ BEC_\varepsilon \ln P $.

To show $F(C_\eta) \ln BSC(C_\eta)$
we use~\cite[Proposition 8]{Makur_2018}, which states that 
 $W \sln V $ if and only if
\begin{align}
    F(P_X)=\chi^2(W\circ P_X||W\circ Q_X) -\chi^2(V\circ P_X||V\circ Q_X)  \nonumber
\end{align}
is convex in $P_X$ for all $Q_X$.
Taking the second derivative, we get,
\begin{align}\label{eq:convex_criterion}
    \frac12 F''(p) &= \sum_{y>0} (w_y+w_{-y}) \frac{\frac{(w_y - w_{-y})^2}{(w_y + w_{-y})^2}}{q\ast\frac{w_y}{w_y+w_{-y}}(1-q\ast\frac{w_y}{w_y+w_{-y}})} \nonumber\\&- \sum_{y>0} (v_y+v_{-y}) \frac{\frac{(v_y - v_{-y})^2}{(v_y + v_{-y})^2}}{q\ast\frac{v_y}{v_y+v_{-y}}(1-q\ast\frac{v_y}{v_y+v_{-y}})}
\end{align}
We want to choose $V$ as a BSC. Hence, the above simplifies to, 
\begin{align}
    \frac12 F''(p) &= \sum_{y>0} (w_y+w_{-y}) \frac{\frac{(w_y - w_{-y})^2}{(w_y + w_{-y})^2}}{q\ast\frac{w_y}{w_y+w_{-y}}(1-q\ast\frac{w_y}{w_y+w_{-y}})} \nonumber\\&-  \frac{(1-2p)^2}{q\ast p (1-q\ast p)}
\end{align}
More specifically, we want to choose $p$ such that the BSC has the same contraction coefficient as $W$. The contraction coefficient of the BSC is, 
\begin{align}
    \eta_{\chi^2}(BSC(p)) = (1-2p)^2. 
\end{align}
Hence $p=\frac12(1\pm \sqrt{\eta})$ and $p(1-p)=\frac14(1-\eta)$ where $\eta\equiv \eta_{\chi^2}(W)$. Furthermore, we use that
\begin{align}
    a\ast b (1-a\ast b) = b(1-b) + (1-2b)^2 a(1-a). 
\end{align}
With that, we can write $\frac12 F''(p)$ as
\begin{align}
   & \sum_{y>0} (w_y+w_{-y}) \frac{\frac{(w_y - w_{-y})^2}{(w_y + w_{-y})^2}}{q(1-q)+\frac14(1-2q)^2\left(1-\frac{(w_y - w_{-y})^2}{(w_y + w_{-y})^2}\right)} \nonumber\\&-  \frac{\eta}{q(1-q)+\frac14(1-2q)^2(1-\eta)}
\end{align}
From the closed formula of the contraction coefficient, it is now easy to see that $F''(p)\geq0$ if the function
\begin{align}
    f(x)=\frac{x}{q(1-q)+\frac14(1-2q)^2\left(1-x\right)}
\end{align}
is convex. 
For that, we check,
\begin{align}
    f''(x) = \frac{8(1-2q)^2}{\left(1-(1-2q)^2 x \right)^3}, 
\end{align}
which can be seen to be non-negative because $q,x\in[0,1]$.
\end{proof}

In addition, using the same criterion as above, we see that BISO channels in $C_\eta$ of output dimension at most three are less noisy comparable.  
\begin{lemma}\label{lemma:lncomp}
    Any two BISO channels $F(C_\eta), G(C_\eta)$ of output dimension at most three are less noisy comparable
\end{lemma}
\begin{proof}
    As BISO channels with output dimension two are exactly BSC, if either channel has output dimension two, the statement follows from Theorem \ref{th:eta}. In the case that the output dimension is three, we denote the transition probabilities of $F(C_\eta)$  resp. $G(C_\eta)$ by $\{p_0, p_{\pm 1}\}$ resp. $\{w_0, w_{\pm 1}\}$. The sum in Equation~\eqref{eq:convex_criterion} then reduces to only one term and we can cancel out $\eta$ on both sides. $  G(C_\eta)  \ln F(C_\eta) $ is equivalent to 
\begin{align}
    &q \ast \frac{p_1}{1-p_0}\left(q \ast \frac{p_{-1}}{1-p_0}\right) \geq q \ast \frac{w_1}{1-w_0}\left(q \ast \frac{w_{-1}}{1-w_0}\right)\\ 
  &\Leftrightarrow 
    \frac{w_1w_{-1}}{(1-w_0)^2} \geq \frac{p_1p_{-1}}{(1-p_0)^2}.
\end{align}
Hence, the less noisy comparability amounts to comparing two positive numbers.
\end{proof}

\begin{remark}\label{rem:lnCounter}
   For output dimension four, we can find a counterexample. The two channels with transition probabilities $\{p_y\} =\{0.01, 0.48, 0.32, 0.19\}$, respectively $ \{q_y\} = \{0.3- \frac{17}{997},\frac{17}{997}, 0, 0.7 \}$, satisfy both $\eta = 0.194 $. Yet, for $q = 0.001$, $F''(p) \approx -14.4 <0$ while at $q= 0.02$, $F''(p) \approx0.9 >0$.
\end{remark}

Now we consider the extremes in $C_\alpha$ with respect to the degradable partial order. 
\begin{theorem}\label{th:alpha}
    For any BISO channel $F(C_\alpha)$, we have
    \begin{align}
        BEC(C_\alpha) \deg F(C_\alpha) \deg BSC(C_\alpha). 
    \end{align}
\end{theorem}
\begin{proof}
$  BEC(C_\alpha) \deg F(C_\alpha)$ follows directly by the characterization of the Doeblin coefficient via the erasure channel in Definition \ref{eq:Doeblin_E}. For $p \leq \frac{1}{2}$, $\alpha(BSC(p)) = 2p$, hence $p =  \sum_{y>0}\min\{p_y, p_{-y}\}  
$. 
For $F(C_\alpha) \deg BSC(C_\alpha)$ the channel $A$ with entries defined for $ y \in \{-l, \dots,l\}$ as 
\begin{align}
    a_y = \begin{cases}
    1 &p_y\geq p_{-y}\\
    0 & \text{else}
\end{cases} = \mathbb{I}(p_y\geq p_{-y})
\end{align}
satisfies $BSC(C_\alpha) = A \circ F(C_\alpha)$ because it satisfies the resulting equations:
\begin{align}
    \sum_{y>0} p_ya_y+p_{-y}a_{-y} = \sum_{y>0} \max\{p_y, p_{-y}\},\\
    \sum_{y>0} p_ya_{-y} + p_{-y} a_y = \sum_{y>0} \min\{p_y, p_{-y}\}.
\end{align}
\end{proof}

The comparability of BISO channels in $C_\alpha$ with output dimension at most three w.r.t. the degradable partial order is also guaranteed. 

 \begin{lemma} \label{lem:alphaComp}
    Any two BISO channels $F(C_\alpha), G(C_\alpha)$ of output dimension at most three are comparable w.r.t the degradability partial order.
\end{lemma}
\begin{proof}
   As argued in the proof of Lemma~\ref{lemma:lncomp}, due to Theorem \ref{th:alpha}, we only need to consider output dimension exactly three.
    We denote the transition probabilities of the two channels $F$ resp. $G$ by $\{p_k\}$ resp. $\{q_k\}$. W.l.o.g. $p_0\leq q_0$. Then it holds that $G(C_\alpha) \deg F(C_\alpha)$ i.e. $F=A\circ G$
with 
\begin{align}
    A= \begin{cases}
    \left(\begin{smallmatrix}
        1 &0&0\\
        \frac{p_1-q_1}{q_0} & \frac{p_0}{q_0} & \frac{p_1-q_1}{q_0}\\
        0 & 0 & 1
    \end{smallmatrix}\right) & p_1\geq q_1\\[2.55mm]
    \left(\begin{smallmatrix}
        0 & 0 &1\\
        \frac{p_1- q_{-1}}{q_0} & \frac{p_0}{q_0} & \frac{p_1-q_{-1}}{q_0}\\
        1 & 0 & 0
    \end{smallmatrix}\right) & \text{else}
    \end{cases}.
\end{align}
By Observation \ref{ob:bin}, requiring the same $\alpha$ is equivalent to fixing the total variation distance
between the row probability distributions, hence $\abs{q_1-q_{-1}}=  \abs{p_1-p_{-1}}$. Combining this constraint with the assumption $p_0\leq q_0$, leads to $q_1 \leq p_1$ in case $\sign{p_1-p_{-1}} = \sign{q_1-q_{-1}}$ and to $q_{-1} \leq p_1$ in case $\sign{p_1-p_{-1}} = - \sign{q_1-q_{-1}}$. This implies the non-negativity of the coefficents of the degrading map. In the following, we assume the case $\sign{p_1-p_{-1}} = \sign{q_1-q_{-1}}$, the other case can be proven analogously.
That $A$ is row-stochastic, can then be seen by substituing $p_0=1-p_1-p_{-1}$, resp. $q_0=1-q_1-q_{-1}$, inside the constraint,
\begin{align}
    &p_1-p_{-1} -q_1+q_{-1}=0\\ &\Leftrightarrow 2p_1-2q_1+1-p_1-p_{-1} = 1-q_{-1}-q_1 \\&\Leftrightarrow p_0+2(p_1-q_1) = q_0.
\end{align}
The resulting equations $F=A\circ G$ for $p_{0,1}$ are satisfied trivially, while the equation for $p_{-1} = q_{-1}+p_1-q_1$ is equivalent to the constraint on the total variation distance.
\end{proof}
\begin{remark}\label{rem:alphaCounter}
We can find two channels in $C_\alpha$ of output dimension four that are not degradable into each other: The
channel $F$ with transition probabilities $\{p_y\} = \{0.415, 0.345, 0.05, 0.19\}$, and $G $ with  $\{q_y\} = \{0.245, 0.515, 0.221, 0.019\}$ satisfy $\alpha(F)= \alpha(G) = 0.48$.
To check the degradability, we use its characterization via the min-entropy in Equation~\eqref{eq:deg_min_entropy} and note that
$H_{\rm{min}}(F) = \sum_{y>0} \max\{xp_y, (1-x)p_{-y}\} +
\max\{xp_{-y}, (1-x)p_y\}$, 
where  $X \sim Ber(x) $.
For $x = 0.12$, $H_{\rm{min}}(F) = 0.88< 0.89268= H_{\rm{min}}(G)$ while at $x = 0.29 $, 
 $H_{\rm{min}}(F) =0.775> 0.76756= H_{\rm{min}}(G)$.
\end{remark}

Following the idea of expansion coefficients, reverse coefficients to the Doeblin coefficient were introduced in~\cite{hirche2024quantumdoeblincoefficientssimple}. To this end, the BEC is replaced by a BSC
in Equation~\eqref{Eq:KL-LN-con} and \eqref{eq:Doeblin_E},
\begin{align}
 \widecheck\alpha(P) &= \inf \{2p: P \deg BSC_p\}, \\
\widecheck \beta(P) &= \inf \{4p(1-p): P \ln BSC_p\},\\
\widecheck \gamma(P) &= \inf \{h_2(p): P \mc BSC_p\}.
\end{align}
Note that we choose a slightly different normalization here compared to~\cite{hirche2024quantumdoeblincoefficientssimple}. With the results obtained, we can now relate them to the original coefficients. In~\cite{hirche2024quantumdoeblincoefficientssimple} it was shown that, using the result in~\cite{geng2013broadcast}, we have for BISO channels,
\begin{align}
\widecheck \gamma(P) = \gamma(P) = C(P). 
\end{align}
With our results above, we can extend that result to the other coefficients.  
\begin{proposition}\label{prop:reverse}
For any BISO channel $P$,
\begin{enumerate}
    \item $ \alpha(P) = \widecheck \alpha(P) =  1-\eta_{TV}(P)$
    \item  $\beta(P) =\widecheck \beta(P) = 1-\eta_{KL}(P)$
\end{enumerate}
\end{proposition}
\begin{proof}
    $1)$  
    By Theorem~\ref{th:alpha}, we have $ \alpha(P) \geq 1-\eta_{TV}(P) \geq \widecheck\alpha(P) $, because it provides feasible solutions to the optimization probelems.
    Additionally, for all channels $P$ (not necessarily BISO or binary) by~\cite[Lemma III.2 and Lemma IV.I]{hirche2024quantumdoeblincoefficientssimple} we have $\alpha(P) \leq 1-\eta_{TV}(P) \leq 1-\widecheck\eta_{TV}(P) \leq \widecheck\alpha(P)$.  
    $2)$ follows similarly by Theorem \ref{th:eta}, giving $ \beta(P) \geq 1-\eta_{KL}(P) \geq \widecheck\beta(P) $ for BISO channels, and by definition from which $\widecheck\beta(P) \geq 1-\eta_{KL}(P) \geq \beta(P) $ for all channels.
\end{proof}

\section{Extendability to general binary channels}

In this section, we show that the BSC is exclusively extremal in the above-defined sense for the class of BISO channels and not for general binary input channels.
We denote a binary input and binary output channel by $ D(q,p) =   \begin{pmatrix}
            1-p & p\\
            q & 1-q
        \end{pmatrix}$,
and furthermore, the asymmetric binary channel $D(0,q) $ by $Z$.

First, we note that we can make a similar statement to Theorem~ \ref{th:alpha} for general binary channels.

\begin{proposition}\label{prop:general}
     For any binary channel $F(C_\alpha)$, there exists a binary input and binary ouput channel $D(C_\alpha)$, s.t.
    \begin{align}
        BEC(C_\alpha) \deg F(C_\alpha) \deg D(C_\alpha). 
    \end{align}
\end{proposition}
    \begin{proof}
    $ BEC(C_\alpha) \deg F(C_\alpha) $ follows again by definition.
      First, we show that for any binary channel $F(C_\alpha)$, there is a binary input and binary
output channel $D$ s.t., $F(C_\alpha) \deg D(C_\alpha)$. 
        We denote the transition probabilities of $F(C_\alpha)$ by $r_y = P(Y=y|X=0), s_y = P(Y=y|X=1)$. Define
         $a_y = \begin{cases}
            1 & s_y \leq r_y\\
            0 & \text{else}
        \end{cases} = \mathbb{I}(s_y\leq  r_y)$, $s= \sum_y s_y a_y, r= \sum_y (1-a_y)r_y$.
       W.l.o.g.  assume that $s,r \leq \frac{1}{2}$, then for fixed Doeblin coefficent $\alpha(D(r,s)) = r +s= \sum_y \min\{s_y, r_y\}$. $A$ with row probabilities $a_y $ and $1-a_y$ then satisfies $ D = A\circ F$. 
    \end{proof}

\begin{remark}\label{rem:ext}
Note that the above proposition is somewhat different from the previous results in the sense that the target channel $D(C_\alpha)$ depends on $F(C_\alpha)$ itself, not just $\alpha(F)$. This is perhaps not surprising, as even two binary input binary output channels are generally not comparable. To see this assume $D(p',q') = D(a,b)\circ D(q,p)$ which requires $p'=(1-p)a+p(1-b)$ and $q'=q(1-a)+(1-q)b$. Suming both equations, and using $p+q=p'+q'$, this can only hold if either $1-p=q$ or $a+b=0$, the former gives a different symmetry and the latter implies that the original and the target channels are the same. 
\end{remark}
For the classes $C_\eta$ and $C_C$ the BSC is also not extremal when we consider general binary input channels. 
\begin{example}\label{ex:Zeta}
 $Z(C_{\eta})\not\ln BSC(C_\eta)$.
 Using equation \ref{eq:eta_bin},
\begin{align}
    \eta(Z(q)) &= \sup_x \frac{(1-q)^2 (1-x)x}{x+(1-x)q}+ (1-q)x \\&= \sup_x \frac{(1-q)x}{q+(1-q)x} = 1-p.
\end{align}
As the argument monotonically increases in $x$, it is maximized for $x=1$.
Hence $\eta(Z(4p(1-p))) = \eta(BSC(p))$.
To evaluate Eq.~\eqref{eq:convex_criterion},
\begin{align}
     &\chi^2(Z(r)\circ P||Z(r)\circ Q) \\
     &=\frac{(p-q)^2(1-r)^2}{q+(1-q)r} + \frac{(p-q)^2(1-r)}{(1-q)} \\
     &= (p-q)^2 \frac{1-r}{(1-q)(q+(1-q)r)}\\
     & =  (p-q)^2 \frac{\eta}{(1-q)(q+(1-q)(1-\eta))} .
\end{align}
Assuming $Z(C_{\eta})\ln BSC(C_\eta)$ or equivalently $F''(P_X) \geq 0$, results in 
\begin{align}
   (1-q)(q+(1-q)(1-\eta)) \leq q(1-q)+\frac{(1-2q)^2(1-\eta)}{4}. 
\end{align}
However, for $q< \frac{3}{4} $, the inequality is violated, and therefore, $Z(4p(1-p))$ and $ BSC(p) $ are not less noisy comparable.
\end{example}

\begin{example} \label{ex:ZC}
$
Z(C_C)\not\mc BSC(C_C)$.
   Both channels have the same capacity if $C= 1-h_2(p) = \log_2(1+2^{-\frac{h_2(q)}{1-q}})$.
   \begin{align}
    I(X:Y_{BSC_p}) &= h_2(x\ast p) -h_2(p),\\
    I(X:Y_Z) &= h_2(x+(1-x) q) -(1-x) h_2(q)  .
\end{align}
The difference in mutual informations $  I(X:Y_Z) -I(X:Y_{BSC_p})$ evaluates to
\begin{align}
   &h_2(x+(1-x)q)
   -(1-x) h_2(q)- \log_2(1+2^{- \frac{h_2(q)}{1-q}})
   \nonumber\\&-  h_2\left(x \ast 
    h_2^{-1}\left( 1- \log_2(1+2^{- \frac{h_2(q)}{1-q}})   \right)
   \right)+  1.
\end{align}
Numerically, we find this to be non-positive for $q,x \in [0,1]$, hence $BSC(C) \mc Z(C) $.
\end{example}

\section{Applications} \label{sec:app}

In this section, we provide a number of simple applications to the results shown above.

\subsection{Wiretap channels}
In the wiretap channel model, a channel is used as a cryptographic resource \cite{6772207}.
The secrecy capacity is the maximum rate at which
information can be transmitted over a channel while keeping the communication
secure from wiretappers and satisfies $C_s = \max_{P_X} \left(I(X:Y) - I(X:Z)\right)$. By~\cite[Theorem 3]{van1997special}, because for all BISO channels $I(X:Y)$ is maximized by the uniform distribution, if $P_{Y|X}\ln P_{Z|X}$, then $C_s=C(P_{Y|X})-C(P_{Z|X})$.
With that, we can apply Theorem \ref{th:alpha} to find the secrecy capacity if one channel is a BISO channel and the other is the BEC or the BSC with the same contraction coefficient, 
\begin{align}
    C_s( BEC({1-\eta}), W) &= \eta - C(W),\\
    C_s\left(W, BSC\left({\frac{1- \sqrt{\eta}}{2}}\right)\right) &= C(W) - 1+ h_2\left(\frac{1- \sqrt{\eta}}{2}\right).
\end{align}
The capacity of a BISO channel is given as in \cite{geng2013broadcast} by $C(W) = 1 - \sum_{y>0} (p_y+p_ {-y}) h_2(\frac{p_y}{p_y+p_{-y}}) $.

\subsection{Channel $f$-divergences}
For essentially any quantity that obeys data processing, we can give bounds achieved by the extremal channel. 
As an example, we take general $f$-divergences. This are given for a convex function $f$ with $f(1)=0$ by
\begin{align}
    D_f(P_X\|Q_X) = \sum_x Q_{X=x} f\left(\frac{P_{X=x}}{Q_{X=x}}\right).
\end{align}
We can now bound the maximum output divergence of a given channel in terms of, say, its maximal leackage. 
\begin{lemma}\label{lem:fdiv}
    For any BISO channel $P_{Y|X}\in C_\alpha$  with $0<\alpha(P_{Y|X})<1$, we have, 
    \begin{align}
        &\frac{e^{\cL}}{2} f\left(\frac{2-e^\cL}{e^\cL}\right) + \left(1-\frac{e^{\cL}}{2}\right) f\left(\frac{e^\cL}{2-e^\cL} \right) \\
    &\leq\sup_{P_X,Q_X} D_f(P_Y \| P'_Y) \\ &\leq (e^\cL-1) f(0), 
    \end{align}
    where $P_Y=P_{Y|X}\circ P_X$, $P'_Y=P_{Y|X}\circ Q_X$ and $\cL\equiv\cL(P_{Y|X})$. 
\end{lemma}
\begin{proof}
    First, we have, 
    \begin{align}
        \sup_{P_X,Q_X} D_f(P_Y \| P'_Y) &= \sup_{x,x'} D_f(P_{Y|X=x} \| P'_{Y|X=x'}) \\
        &= D_f(P_{Y|X=0} \| P'_{Y|X=1}), 
    \end{align}
where the first line follows from convexity and the second because $P_{Y|X}$ is binary. The result then follows from Theorem~\ref{th:alpha}, data processing and Observation~\ref{ob:bin}.
\end{proof}
Compare also the result in \cite{Issa2022maxleakage}, stating that the Chernoff information is maximized given fixed maximum leagake by the BEC and minimized by the BSC.

\subsection{Contraction with input constraints}

Often, contraction coefficients do not provide the full picture, especially in cases where there is an additional constraint on the input distribution. To study this setting,~\cite{du2017strong} introduced the \textit{best possible data processing function} $F_I$, defined by, 
\begin{align}
    F_I(P_{Y|X},t) &= \sup\{ I(W:Y) : I(W:X)\leq t, W\rightarrow X\rightarrow Y \} ,
\end{align}
see also~\cite{polyanskiy2016,polyanskiy2015dissipation}.

By Theorem~\ref{th:eta}, we can bound the $F_I$-curve for BISO channels $W$ by 
\begin{align}
 1-&h_2\left(  \frac{1\pm \sqrt{\eta}}{2}  \ast h_2^{-1}(\max(1-t,0)) \right) \\
 &\leq F_I(W, t) \\
 &\leq \eta \min\{t,1\},
\end{align}
where $\eta\equiv\eta_{KL}(W_{Y|X})$. Here we used that the $F_I$-curve was evaluated for BSC and BEC channels in~\cite{du2017strong}. Note that the upper bounds also holds for general channels.

\textbf{Author contributions.}
CH formulated the research problem and supervised the project. Both authors
contributed to the formal analysis and the proof of the first main theorem.
OS led the proof of the second main theorem, developed the auxiliary results and examples,
and wrote the first draft. Both authors reviewed and edited the manuscript.
\bibliographystyle{IEEEtran}
\bibliography{lib}

\end{document}